\documentclass{asl}

\usepackage[T1]{fontenc}
\usepackage[utf8]{inputenc}
\usepackage{xspace,prooftree}
\usepackage{a4wide}
\usepackage{xcolor}
\usepackage{url}
\newcount\hour \hour=\time
\newcount\minute \minute=\hour
\divide\hour by 60
\newcount\temp \temp=\hour
\multiply\temp by 60
\advance\minute by -\temp
\newcount\temp \temp=\minute
\divide\temp by 10
\def\optzero{\ifcase\temp 0\fi}

\def\miesiaca{\ifcase\month\or stycznia\or lutego\or marca\or kwietnia\or maja
                           \or czerwca\or lipca\or sierpnia\or wrze\'snia
                           \or pa\'zdziernika\or listopada\or grudnia\fi}

\usepackage{enumitem}\usepackage{a4wide}
\usepackage{proof}
\usepackage{graphicx}
\newtheorem{theorem}{Theorem}
\newtheorem{lemma}[theorem]{Lemma}

\newtheorem{example}[theorem]{Example}

\newtheorem{proposition}[theorem]{Proposition}
\newtheorem{remark}[theorem]{Remark}

\newcommand\ciut{\hspace{.3mm}}

\def\<{\langle\,}
\def\>{\,\rangle}
\def\To{\Rightarrow}\def\Ot{\Leftarrow}
\def\oto{\leftrightarrow}
\newcommand{\rarrow}{\twoheadrightarrow}
\newcommand\FV[1]{\mbox{\rm FV}(#1)\xspace}
\def\pusty{\varnothing}

\definecolor{darkbrown}{rgb}{0.4, 0.26, 0.13}
\newcommand\pspace{\mbox{\sc Pspace}\xspace}
\newcommand\np{\mbox{\sc NP}\xspace}
\newcommand\conp{\mbox{{\it co}-\sc NP}\xspace}
\newcommand\ipc{\mbox{\rm IPC}\xspace}
\newcommand\iipc{\mbox{\rm IIPC}\xspace}
\newcommand\tg[1]{\mathrm{tg}(#1)}

\newcommand{\A}{\mathcal{A}\xspace}

\newcommand{\M}{\mathcal{M}}
\newcommand{\I}{\mathcal{I}}
\newcommand{\C}{\mathcal{C}}

\def\wrt{with respect to\xspace}
\def\iff{if and only if\xspace}

\newcommand{\czek}[3]{\mbox{\tt check}\;#1;\,\mbox{\tt set}\;#2;\,{\tt
    jmp}\;#3}

\newcommand{\jmpone}{{\mathtt{jmp}\;}\xspace}
\newcommand{\jmp}[2]
{\jmpone#1\;\mbox{\tt and}\;#2}

\def\przypadki#1#2#3{\left\{\begin{array}{ll}
#1,&\mbox{if~#2;}\\
                              #3,&\mbox{otherwise}.\end{array}\right.}

\newcommand\TG[1]{{\sf TG}(#1)}
\newcommand\trejs{{\it tr\/}}

\def\forces{\Vdash} \def\notforces{\nVdash}
\newcommand\lip{{\tt p}}\newcommand\liq{{\tt q}}
\newcommand\lir{{\tt r}}
\newcommand\nip{\overline{\tt p}}

\newcommand{\troche}{\hspace{.13332em}}
\renewcommand\:{\troche{:}\troche}
\def\przypadki#1#2#3{\left\{\begin{array}{ll}
#1,&\mbox{if~#2;}\\
#3,&\mbox{otherwise}.
\end{array}\right.}

\xdefinecolor{grin}{rgb}{0.1,0.5,0.1}

\newif\ifmydraft
\mydrafttrue

\xdefinecolor{innyblue}{rgb}{0.2,0.1,0.9}
\newcommand{\rank}{order\xspace}

\title{Between proof construction and SAT-solving}
\author{Aleksy Schubert}
\revauthor{Schubert, Aleksy}
\address{%
  Faculty of Mathematics, Informatics and Mechanics\\
  University of Warsaw\\
  ul. Stefana Banacha 2\\
  02-097 Warsaw\\
  Poland
}
\email{alx@mimuw.edu.pl}
\thanks{The work was partly supported by IDUB
POB 3 programme at the University of Warsaw.}

\author{Paweł Urzyczyn}
\address{%
  Faculty of Mathematics, Informatics and Mechanics\\
  University of Warsaw\\
  ul. Stefana Banacha 2\\
  02-097 Warsaw\\
  Poland
}
\email{urzy@mimuw.edu.pl}

\author{Konrad Zdanowski}
\address{%
  Cardinal Stefan Wyszyński University in Warsaw\\
  ul. Dewajtis 5\\
  01-815 Warsaw\\
  Poland
}
\email{k.zdanowski@uksw.edu.pl}


\begin{document}

\begin{abstract}
 The classical satisfiability problem (SAT) is used as a natural and
 general tool to express and solve combinatorial problems that are
 in~\np.
  We postulate that provability for implicational intuitionistic
  propositional logic
  (\iipc)
  can serve as a similar natural tool to
  express problems in \pspace.
  We demonstrate it by proving two essential results concerning
  the system.
  One is that provability in full
  \ipc
  (with all connectives) can be reduced to provability of implicational
  formulas of order three. Another 
  result is a convenient interpretation
  in terms of simple alternating automata.
  Additionally, we distinguish some natural subclasses of
  \iipc
  corresponding to the complexity classes
	\np and \conp.
\end{abstract}

\maketitle

\section{Introduction}
\label{sec:intro}

Everyone knows 
that 
classical propositional
calculus (CPC) is a~natural language to represent 
combinatorial problems
(see e.g.\ \cite{HoosStutzle05,Roig13}). Various decision problems can be 
easily encoded as instances of the 
formula satisfiability problem (SAT)
and efficiently solved \cite{AlounehASM19,PrasadBG05}.
%

In this paper we would like to turn the reader's attention to the so
far unexploited fact that intuitionistic implicational propositional
calculus (\iipc) \cite{Johansson36} is an interesting propositional
formalism which \relax 
is equally natural and simple in its nature as CPC, yet
stronger in its expressive power. Indeed, while SAT and ASP~\cite{Brewkaetal}
can express \np-complete problems, the decision problem for \iipc is
complete for \pspace. Thus \iipc can \relax{accommodate} a larger class of
problems that may be encoded as formulas and solved using automated or
interactive proof-search. In particular, the Sokoban 
puzzle~\cite{Culberson97,dorzwick99,hearndemaine05}
cannot be solved by means of 
SAT solving, but could be 
encoded in \iipc and examined by an interactive prover. \relax 

Of course the \pspace complexity is
enormous, 
but the general case of \np is \relax{infeasible} anyway. 
And not every polynomial
space computation 
requires exponential time.
We may only solve
``easy'' cases of hard problems, and then the increased expressiveness
of the language can be useful rather than harmful. For example, since
\pspace is closed under complements one can simultaneously attempt to
prove a proposition and to disprove it by proving a dual
one~\cite{urzy-types20}.

What is also important, this approach to \pspace avoids 
adding new syntactical forms such as \relax{Boolean} quantification
of QBF \cite{StockmeyerMeyer73}. 
Moreover, we can 
syntactically distinguish 
subclasses of \iipc 
for which the decision problem is complete for
P, \np and \conp. 

The strength of CPC and SAT-solving is in its conceptual simplicity --
a propositional formula provides a specification of a configuration of
interest while a solution provides a particular configuration that
meets the specification. In the 
case of \iipc, as illustrated below,
we are able to achieve the same
goal. In addition, we obtain one more 
dimension  of expressiveness:  
the proof we build represents the process of constructing the solution.
For instance, a~sequence of moves in the Sokoban game, or a computation
of a machine corresponds to a~sequence of proof steps (in the order
of which the proof is being constructed). 

Indeed, 
interestingly enough, \iipc offers not only the formalism to describe
relationships, but also has a 
procedural view in the form of
proof-search process. Moreover, the proof-search in \iipc does not have
to be less convenient than processing %
SAT instances or
computing in ASP-based paradigm~\cite{Brewkaetal}: normal proof
search (Ben-Yelles algorithm) is intuitive and straightforward,
especially because one can restrict attention to formulas of \relax
order at most three.

The proof-search computational paradigm brings here an interesting,
not always clearly expressed, facet to the Curry-Howard
isomorphism. The Curry-Howard isomorphism states that systems of
formal logic 
and computational calculi 
are in direct correspondence. The two most 
commonly recognised facets  of the correspondence are expressed by
the slogans ``formulas as types'' and ``proofs into programs''.
The third facet %
of the isomorphism brings in the
observation that normalisation of a proof 
is the same as the computation of the corresponding 
program. 
The fourth and frequently overlooked 
side of Curry-Howard is the computational content
of proof construction. Virtually any algorithm can be expressed 
as a~formula of some constructive logic in such a way that 
every proof of the formula is but an execution of the algorithm. 
Yet differently: finding a proof of a~formula (or equivalently
an inhabitant of a~type) is the same as executing a program. 
This way we have a close relationship between 
\emph{proof search} in the realm of logic or \relax
\emph{program synthesis} \cite{kupferman-vardi-bsl99, rehof-vardi14}
in the realm of programming. 

A simple illustration of the paradigm ``proof construction as computation'' 
is reading a logical consequence $\Gamma\vdash\tau$ as a configuration of
a machine (a {\it monotonic automaton\/}).  
Under this reading the proof goal is the internal state, \relax
the assumptions $\Gamma$ represent the memory. Variants of such 
{\it monotonic automata\/} were used in~\cite{aleksybar,trudne}; 
in the present paper we make this automata-theoretic semantics
of (I)\ipc very clear-cut. 

We begin our presentation with Section~\ref{sec:praries} where we fix
notation and recall some basic definitions. 
Then we enter the discussion 
of expressibility of \iipc,
focusing mainly on the fact that the whole expressive strength is in
formulas of order at most~three. In Section~\ref{monaut} %
we demonstrate 
the natural equivalence between proof-search and computation:
 the monotonic automata directly implement the Wajsberg/Ben-Yelles 
inhabitation algorithm for the full \ipc (with all connectives). 
After showing that the
halting problem for monotonic automata is
\relax{\pspace-complete,}
we reduce it to provability in \iipc. 
This yields a~polynomial reduction of the decision problem 
for the full \ipc 
to \iipc formulas of order at most~three. 
It follows from Section~\ref{sec:hierarchy} however, that
the translation does not preserve the equivalence of formulas. \relax
Still our reduction plays a~similar role as that of the 
whole SAT to 3-CNF-SAT. 

In Section~\ref{sec:low-rank} we define two subclasses of low-order 
\iipc
corresponding to the 
complexity classes
\np
and \conp.

We conclude in Section~\ref{sec:conclusions} with a few final comments.

\section{Preliminaries}
\label{sec:praries}

To make the paper self-contained, we introduce here the necessary 
definitions  \relax
and fix the basic notation. This section may be to large
extent skipped and used as a reference. 
A more detailed account of the relevant \relax
notions can be found for instance in \cite{SU06}.

\newcommand{\tyv}{\mathcal{X}\xspace}
\newcommand{\tev}{\Upsilon\xspace}
\newcommand{\FTV}[1]{\mathrm{Atm}(#1)\xspace}
\newcommand{\NF}[1]{\mathrm{NF}(#1)\xspace}
\newcommand{\dom}[1]{\mathrm{dom}(#1)\xspace}
\newcommand{\abs}[1]{|#1|\xspace}
\newcommand{\Cycle}{\mathrm{Cycle}\xspace}
\newcommand{\p}{p\xspace}
\newcommand{\q}{q\xspace}
\newcommand{\Redundant}{\mathcal{R}\xspace}
\newcommand{\vdashstlc}{\vdash_{\hspace{-3pt}{}_\to}\xspace}
\newcommand{\vdashp}{\vdash_{p}\xspace}
\newcommand{\vdashsr}{\vdash_{sr}\xspace}
\newcommand{\vbracket}[1]{\Vert #1\Vert\xspace}
\newcommand{\abracket}[1]{\lceil #1\rceil\xspace}
\newcommand{\parfunc}{\rightharpoonup\xspace}
\newcommand{\Types}{\mathbb{T}\xspace}
\newcommand{\Formulas}{\mathbb{F}\xspace}
\newcommand{\Ctx}{\mathrm{Ctx}\xspace}
\newcommand{\inn}{\mathrm{in}\xspace}

\subsubsection*{Propositional formulas}

We assume an infinite set $\tyv$ of {\it propositional variables\/},
usually written \relax
as $p,q,r,\dots$, possibly with indices. Propositional variables and
the constant~$\bot$ are called {\it atoms}.

The formulas of the full intuitionistic propositional
logic, \ipc, are generated by the grammar:
$$
\varphi,\psi::=p\mid \bot \mid \varphi\to \psi \mid \varphi\land\psi 
\mid  \varphi\lor\psi,
$$
where $p\in\tyv$. 
As usual, we use $\lnot\varphi$ as a shorthand for $\varphi\to\bot$.

For clarity we do not include parentheses in the grammar. We adopt 
standard conventions that parentheses can be used \relax
anywhere to disambiguate understanding of the formula structure.
Additionally, we assume 
that $\to$ is right-associative so
that $\varphi_1\to \varphi_2\to \varphi_3$ is equivalent to
$\varphi_1\to (\varphi_2\to \varphi_3)$. 

We use the notation $\varphi[p:=\psi]$ for substitution. 
If
\relax{$\Gamma = \{ \varphi_1,\ldots,\varphi_n\}$}
then we write
\relax{$\Gamma\to p$}
for the formula
\relax{$\varphi_1\to\cdots\to\varphi_n\to p$.}


A \emph{literal} is either 
a~propositional variable or a negated variable.
Literals are written in typewriter font: $\lip, \liq, \lir,\dots$
If $\lip$ is a literal, then $\nip$ is its
dual literal,\relax
 i.e.~$\overline p=\neg p$ and $\overline{\neg p}=p$. 


\subsubsection*{Proof terms}
According to the Curry-Howard correspondence, formulas can be seen
as types assigned to proof terms. 
In this view, \iipc is exactly the ordinary simply typed lambda-calculus.
For the full \ipc we need an extended calculus and we now define 
the syntax of it. 
 We assume an infinite set
$\tev$ of {\it proof \relax 
variables\/}, usually written as $x, y, z,\dots$ with possible
annotations. A~\emph{context} is a finite set of pairs $x:\varphi$,
where $x$ is a proof 
variable and $\varphi$ is a formula, such that no proof variable occurs twice.
Contexts are traditionally denoted by $\Gamma,\Delta$, etc. 
 If this does not lead
to confusion we identify contexts with sets of formulas (forgetting 
the proof variables). 

We define the Church style (raw) \relax
terms of intuitionistic propositional logic
as follows: 
%
\begin{displaymath}
\begin{array}{l@{\;}l}
  M,N ::= & x\mid
             M[\varphi] \mid 
            \lambda x\:\varphi.\,M \mid
            MN \mid
            \<M,N\> \mid
            M\pi_1 \mid
            M\pi_2 \\ & \mid
            \inn_1 M \mid
            \inn_2 M \mid
            M[x_1\:\varphi.\,N_1;\;x_2\:\psi.\,N_2]
\end{array}  
\end{displaymath}
where $x,x_1,x_2\in\tev$ and $\varphi,\psi\in\Formulas$.  The set of
$\lambda$-terms generated in this way is written $\Lambda^p$.
Again we do not include
parentheses in the grammar, but they can be used anywhere to
disambiguate parsing.  As a~shorthand for
$\lambda x_1\:\sigma_1\ldots \lambda x_n\:\sigma_n.\,M$ we
write $\lambda x_1\:\sigma_1 \ldots x_n\:\sigma_n.\,M$.  
In case this does
not lead to confusion, we omit type annotations from terms and write
for example 
$\lambda x.\,M$ instead of $\lambda x\:\varphi.\,M$ or
$M[x_1.\,N_1;\;x_2.\,N_2]$ for $M[x_1\:\varphi.\,N_1;\;x_2\:\psi.\,N_2]$. 
We also use the common convention that application is left-associative:
$MNP$ stands for $(MN)P$. 
We often write 
e.g.~$ME$ not only for 
application of $N$ to a term~$E$ but also for any elimination:
$E$~can be a projection $\pi_1$ or $\pi_2$, or a $\vee$-eliminator
$[x\:\varphi.\,P;\;y\:\psi.\,Q]$ or a~\mbox{$\bot$-eliminator}~$[\varphi]$.

\relax

The set of
\emph{free variables} in a term is defined as
\begin{itemize}
\item $\FV{x} = \{ x\}$,
\item $\FV{\lambda x\:\varphi.\,M} = \FV{M}\backslash\{ x\}$,
\item $\FV{MN} = \FV{\<M,N\>} = \FV{M}\cup\FV{N}$,
\item $\FV{M\pi_i} = \FV{\inn_i M} = \FV{M[\varphi]} = \FV{M}$ for $i=1,2$,
\item $\FV{M[x\:\varphi.\,N_1;\;y\:\psi.\,N_2]} =
  \FV{M}\cup(\FV{N_1}\backslash\{x\})\cup(\FV{N_2}\backslash\{y\})$.
\end{itemize} 
As usual the terms are tacitly considered up to $\alpha$-conversion so
that the names of non-free variables are not relevant.
\emph{Closed terms} are terms that have no occurrences of
  free variables.
We use
the notation 
$M[x:=N]$ for capture-free substitution of $N$ for the free 
occurrences of $x$ in~$M$.

 \begin{figure}[hbt]
  \centering
$$
\infer[(var)]{\Gamma, x\:\varphi\vdash x:\varphi}{}
$$

$$
\infer[(\to I)]{\Gamma\vdash\lambda x\:\varphi.\,M:\varphi\to \psi}
{\Gamma, x\:\varphi\vdash M:\psi}
$$

$$
\infer[(\to E)]{\Gamma\vdash M_1M_2:\psi}
{\Gamma\vdash M_1: \varphi\to \psi & \Gamma\vdash M_2: \varphi}
$$

$$
\infer[(\land I)]{\Gamma\vdash\<M,N\>:\varphi\land \psi}
{\Gamma\vdash M:\varphi &
\Gamma\vdash N:\psi}
$$

$$
\infer[(\land E1)]{\Gamma\vdash M\pi_1:\varphi}
{\Gamma\vdash M: \varphi\land\psi}
\qquad
\infer[(\land E2)]{\Gamma\vdash M\pi_2:\psi}
{\Gamma\vdash M: \varphi\land\psi}
$$

$$
\infer[(\lor I1)]{\Gamma\vdash \inn_1 M:\varphi\lor\psi}
{\Gamma\vdash M:\varphi}
\qquad
\infer[(\lor I2)]{\Gamma\vdash \inn_2 M:\varphi\lor\psi}
{\Gamma\vdash M:\psi}
$$

$$
\infer[(\lor E)]{\Gamma\vdash  
M[x\:\varphi.\,N_1;\,y\:\psi.\,N_2]:\rho}
{\Gamma\vdash  M: \varphi\lor\psi & \Gamma,x\:\varphi\vdash N_1: \rho &
\Gamma,y\:\psi\vdash N_2: \rho}
$$

$$
\infer[(\bot E)]{\Gamma\vdash M[\varphi]:\varphi}
{\Gamma\vdash M:\bot }
$$
\caption{Proof assignment in \ipc.}
\label{fig:rules}

\end{figure}

The natural deduction inference rules of \ipc are 
presented in Figure~\ref{fig:rules} in the form of type-assignment system
deriving judgements of the form
$\Gamma\vdash M:\varphi$ 
(read: \mbox{``$M$ has type~$\varphi$ in~$\Gamma$\ciut''}
or ``$M$ proves \mbox{$\varphi$ in~$\Gamma$\ciut''}), 
where $\Gamma$ is a~context and $M$ is 
a~proof term. From time to time we use the simplified notation 
$\Gamma\vdash \sigma$ to state that $\Gamma\vdash M:\sigma$
holds for some~$M$. If $\Gamma$ is known, implicit, or irrelevant 
we can simplify the statement 
$\Gamma\vdash M:\tau$ to $M:\tau$ (read ``$M$ has type~$\tau$'').

\subsubsection*{Reductions}

An introduction-elimination pair constitutes a {\it beta-redex\/},
and we have the following set of beta-reduction rules for all the
logical connectives except~$\bot$:
\begin{displaymath}
  \begin{array}{l@{\;\;}c@{\;\;}l}
    (\lambda x\:\sigma. M) N &\To_\beta & M[x := N],\\
    \<M, N\>\pi_1 &\To_\beta & M,\\
    \<M, N\>\pi_2 &\To_\beta & N,\\
    (\inn_1 M)[x\:\varphi.\,N_1;\;y\:\psi.N_2] 
&\To_\beta & N_1[x:= M] ,\\
    (\inn_2 M)[x\:\varphi.\,N_1;\;y\:\psi.\,N_2] 
&\To_\beta & N_2[y:= M].\\
  \end{array}
\end{displaymath}
Other redexes represent elimination steps applied to
a conclusion of a $\vee$- or $\bot$-elimination. \relax
The following rules, called 
{\it permutations\/} (aka commuting conversions), 
permute the ``bad'' elimination upwards. For the disjunction
there is the following general scheme:
$$
M[x\:\varphi.\,N_1;\;y\:\psi.\,N_2]E ~~\To_p~~
M[x\:\varphi.\,N_1E;\;y\:\psi.\,N_2E],
$$ 
where $E$ is any eliminator, i.e., $E\in \Lambda^p$, or $E\in 
\{\pi_1, \pi_2\}$, $E =[\vartheta]$, or
$E=[z\:\vartheta.\,N;\;v\:\varrho.\,Q]$.

Permutations for $M[\varphi]$ follow the pattern
$$
M[\varphi]E  ~~\To_p~~ M[\psi],
$$
where $\psi$ is the type of~$M[\varphi]E$. For example: 
$$
M[\varphi\To\psi]N \To_p M[\psi].
$$


The relation $\to$ is the contextual closure of rules $\To_\beta$ and 
$\To_p$, and $\rarrow$ stands for the 
reflexive-transitive closure of~$\to$.

The system $\Lambda^p$ has a number of important consistency 
features. \relax 
\bigbreak

\begin{theorem}\label{theorem:basic}
The system $\Lambda^p$ has the following properties: \relax 
  \begin{enumerate}
 \item\label{srbasic} 
Subject reduction: 
if $\Gamma\vdash M:\sigma$ and $M\rarrow N$ then
    $\Gamma\vdash N:\sigma$.
  \item\label{crbasic} 
Church-Rosser property: 
if $M\rarrow N$ and $M\rarrow P$ then
    there is a~term $Q$ such that $N\rarrow Q$ and $P\rarrow Q$.
  \item\label{snbasic}
Strong normalisation: 
every reduction $M_1\to M_2\to\cdots$ is finite.
  \end{enumerate}
\end{theorem}

\begin{proof} Part~\ref{srbasic} can be easily verified by observing that
every reduction rule preserves typing.
Part~\ref{crbasic} follows from general results on higher-order 
rewriting~\cite[Chapter 11.6]{terese}, 
because 
the rules are left-linear and non-overlapping.
For part~\ref{snbasic}, see~\cite{degroote-iandc02}.
\end{proof}

A type $\tau$ is {\it inhabited\/} iff there is a closed term $M$ 
such 
that $\vdash M:\tau$ (an {\it inhabitant\/}). 

\subsubsection*{Long normal forms}

It follows from~Theorem~\ref{theorem:basic}(\ref{snbasic}) that 
every inhabited type has a normal inhabitant. To organize and
narrow proof search it is convenient to use 
a stricter notion of long normal form (lnf). In the lambda-calculus
(or equivalently: in natural deduction) long normal forms play
a role similar to {\it focusing\/}~\cite{MillerNPS91,LIANG20094747} in sequent
calculus.

We say that a term $M$ such that
$\Gamma\vdash M:\varphi$ is in \emph{long normal
  form} when one of the following cases holds:

\begin{itemize}
\item $M$ is a constructor 
  $\lambda x.\, N$, $\<N_1,N_2\>$, $\inn_1 N$, or $\inn_2 N$, where
  terms %
  $N, N_1,$ and $N_2$ are lnf.
\item $M = xE_1\ldots E_n$, where $E_1,\ldots,E_n$ are projections or
terms in long normal form, and $\varphi$ is an atom.
\item $M = xE_1\ldots E_nE$, where $E_1,\ldots,E_n$ are 
projections or terms in long normal form, and $E$ is a $\vee$-
or $\bot$-eliminator, and $\varphi$ is either an atom 
or a disjunction.
\end{itemize}
For example, let
$$
\Gamma=\{x\:\alpha\to p,\, y\:\alpha,\,
z\:\alpha\to\beta\vee\gamma,\,
u_1\:\beta\to p,\, u_2\:\gamma\to p\},
$$
where $p$ is an atom. 
Then $\lambda w\:\alpha.\,xw$ is an lnf of type 
$\alpha\to p$, and $zy[v_1\:\beta.\,u_1v_1;\; v_2\:\gamma.\,u_2v_2]$ 
 is an lnf of type~$p$. Also 
$zy[v_1\:\beta.\,\inn_1{v_1};\; v_2\:\gamma.\,\inn_2{v_2}]$
is an lnf of type $\beta\vee\gamma$, while the mere application
$zy$ is not. 


\begin{lemma}[\cite{games}]\label{trzypotrzy} If $\Gamma\vdash \varphi$, then 
$\Gamma\vdash M:\varphi$, for some long normal form~$M$.
\end{lemma}


\subsubsection*{Kripke semantics}

A Kripke model is a triple of the form
$$
\C = \<C,\leq,\forces\>
$$
where $C$ is a non-empty set, the elements of which are called \emph{states},
$\leq$ is a partial order in $C$ and $\forces$ is a binary relation
between elements of~$C$ and propositional variables. The relation
$\forces$, 
read as \emph{forces}, obeys the standard monotonicity
condition: if $c\leq c'$ and $c\forces p$ then $c'\forces p$. 
Without loss of generality we may 
assume that $C$ is finite,~cf.~\cite{Smorynski73}, 
\cite[Section 3]{VanDalen1986}.

Kripke semantics for \ipc is defined
as follows. If $\C = \<C,\leq,\forces\>$  is a Kripke model then
\begin{itemize}
\item $c\forces \varphi\lor\psi$ \iff $c\forces\varphi$ or
  $c\forces\psi$,
\item $c\forces \varphi\land\psi$ \iff $c\forces\varphi$ and
  $c\forces\psi$,
\item $c\forces \varphi\to\psi$ \iff for all $c'\geq c$
  if $c'\forces\varphi$ then $c'\forces\psi$,
\item $c\forces\bot$ does not hold.
\end{itemize}
We write $c\forces\Gamma$, when $c$ forces all formulas in~$\Gamma$.
And $\Gamma\forces\varphi$ means that $c\forces\Gamma$ implies 
$c\forces\varphi$ for each Kripke model $\<C,\leq,\forces\>$ and 
each $c\in C$.

The following completeness theorem holds (see e.g.~\cite{fitm69}):
\begin{theorem}
For each $\Gamma$ and 
$\varphi$, it holds that
  $\Gamma\vdash\varphi$ \iff $\Gamma\forces\varphi$.
\end{theorem}


\subsubsection*{The implicational fragment}
In this paper we are mostly interested in the implicational fragment \iipc
of \ipc. The formulas of \iipc (also known as simple types) are \relax
defined by the grammar 
$$
\sigma,\tau::=p\mid  \sigma\to \tau,
$$
where $p\in\tyv$. 

Any %
formula in \iipc can be written as 
\mbox{$\sigma = \sigma_1\to\cdots\to \sigma_n\to p$,} where $n\geq 0$, and $p$
is a~type atom. Types $\sigma_1,\ldots,\sigma_n$ are the {\it arguments\/},
and the atom $p$ is called the \emph{target} of $\sigma$, \relax
written $p=\tg{\sigma}$. 

The \emph{\rank}
$r(\sigma)$ %
of an implicational formula 
is defined as follows: an atom is of \rank 0, and the \rank of
\relax{$\sigma\to\tau$}
is the maximum of
\relax{$r(\tau)$ and $r(\sigma)+1$.}
In
other words, if $p$ is an atom, then
$$
r(\sigma_1\to\cdots\to\sigma_k\to p)= 1+\max_ir(\sigma_i).
$$


The restricted set $\Lambda_\to$ of \iipc proof-terms is defined by the 
pseudo-grammar: \relax\relax{}
\begin{displaymath}
\begin{array}{l@{\;}l}
  M,N ::= & x\mid
            \lambda x\:
            \sigma.\,M \mid
            MN.
\end{array}  
\end{displaymath}
The relevant rules in Figure~\ref{fig:rules} are 
$(var), (\to I),$ and $(\to E)$, i.e., the type-assignment rules
of the ordinary simply typed lambda-calculus.

\section{Automata for logic}\label{monaut}
It follows from Lemma~\ref{trzypotrzy} that every
provable formula has a long normal proof. This yields a simple 
proof-search algorithm, which is essentially implicit in~\cite{Wajsberg38},
and
hence
called the Wajsberg algorithm
(WA).\footnote{See~\cite{Bezhanishvili1987} for a correction of
the proof in~\cite{Wajsberg38}.}

To present the algorithm we first define the set~$\TG\varphi$ 
of {\it targets of~$\varphi$\/}. Targets are always atoms or disjunctions.


\begin{itemize}[noitemsep,nolistsep]
\item $\TG{\bf a}=\{{\bf a}\}$, when ${\bf a}$ is an atom
(a propositional variable or $\bot$).
\item $\TG{\tau\to\sigma}=\TG\sigma$.
\item $\TG{\tau\vee\sigma}=\{\tau\vee\sigma\}$.
\item $\TG{\tau\wedge\sigma}=\TG\tau\cup\TG\sigma$.
\end{itemize}
Clearly, $\TG{\varphi}=\{\tg\varphi\}$, when $\varphi$ is an implicational
formula.
%

We define
the family $\trejs(\alpha,\varphi)$ 
of  {\it traces to~$\alpha$ in~$\varphi$\/}. Each trace is a~set of formulas.

\begin{itemize}[noitemsep,nolistsep]
\item $\trejs(\alpha,\varphi)=\pusty$ if $\alpha \not\in \TG\varphi$.
\item $\trejs(\alpha,\alpha)=\{\pusty\}$.
\item $\trejs(\alpha,\tau\to\sigma)= \{\{\tau\}\cup T\ |\ 
T\in \trejs(\alpha,\sigma)\}$. 
\item $\trejs(\alpha,\tau\wedge\sigma)=
  \trejs(\alpha,\tau)\cup\trejs(\alpha,\sigma)$.
\end{itemize}

For example, $\trejs(p,r\to(p\wedge(q\to p)\wedge(s\to p\vee q))=
\{\{r\},\{r,q\}\}$. 



\begin{lemma}\label{trzyipolnazad}
  Let $(x\:\psi)\in\Gamma$
  and
  $T\in \trejs(\alpha,\psi)$. If $\Gamma\vdash\rho$, for all
  $\rho\in T$, then $\Gamma\vdash xE_1\ldots E_n:\alpha$,  where
  $n\geq 0$ and $E_1,\ldots, E_n$ are some terms or projections.
\end{lemma}

\begin{proof} Induction \wrt~$\psi$. For example,
  assume 
$\psi=\psi_1\wedge\psi_2$, and  let $T\in \trejs(\alpha,\psi_1)$.
Given that we apply the induction hypothesis to obtain
$\Gamma, y\:\psi_1\vdash yE_1\ldots E_n:\alpha$, where $n\geq 0$,
so 
$\Gamma\vdash x\pi_1E_1\ldots E_n:\alpha$.
\end{proof}

\begin{lemma}\label{trzyipoltam}
Assume that $(x\:\psi)\in\Gamma$ 
and $\Gamma\vdash xE_1\ldots E_m:\varphi$, 
where $E_1,\ldots, E_m$ are terms or projections and $\varphi$ is an
atom or a disjunction. Let $J=\{j\ |\ E_j\ \mbox{is a term}\}$
and let $\Gamma\vdash E_j:\sigma_j$, for all $j\in J$. 
Define $T=\{\sigma_j\ |\ j\in J\}$. 
Then $\varphi\in\TG\psi$, and 
$T\in\trejs(\varphi,\psi)$.
\end{lemma}
\begin{proof} Induction \wrt $m$. 
For example, if $\psi=\psi_1\to\psi_2$, then we apply 
the induction hypothesis to 
$\Gamma,y\:\psi_1\vdash yE_2\ldots E_m:\varphi$.
This implies $T'=\{\sigma_j\ |\ j\in J\wedge j>1\}\in 
\trejs(\varphi,\psi_2)$,
and consequently 
$T=T'\cup\{\psi_1\}\in 
\trejs(\varphi,\psi)$.
\end{proof}

For a given judgement $\Gamma\vdash \varphi$, the \emph{Wajsberg 
algorithm} (WA) 
attempts (implicitly) 
to construct a~long normal proof term. It proceeds as 
follows:



\begin{enumerate}[noitemsep,nolistsep]
\item\label{uwdiUDG} If $\varphi = \tau\wedge\sigma$, call $\Gamma\vdash\tau$ 
and  $\Gamma\vdash\sigma$.
\item\label{asudOFH}
 If $\varphi = \tau\to\sigma$, call $\Gamma,\tau\vdash\sigma$.
\item If $\varphi$ is
  an atom or a disjunction,
  choose $\psi\in\Gamma$ and
  $\alpha\in\TG\psi$ such that
        either $\alpha$ is a~disjunction, or $\alpha=\bot$, 
or $\alpha$ is a~propositional variable
and $\alpha=\varphi$.
Then choose $T\in \trejs(\alpha,\psi)$, and: 
      \begin{itemize}[noitemsep,nolistsep] 
      \item Call $\Gamma\vdash\rho$, for every $\rho\in T$;
      \item If $\alpha=\beta\vee\gamma$, call $\Gamma,\beta\vdash\varphi$
and $\Gamma,\gamma\vdash\varphi$ in addition.
\end{itemize}
\label{en:target}
\end{enumerate}
The procedure accepts in case~(\ref{en:target}), when $\varphi$ is an atom
in~$\Gamma$, as there is nothing to call.

With respect to \iipc case~(\ref{uwdiUDG}) 
disappears and case~(\ref{en:target}) simplifies to
\begin{itemize}
\item[\ref{en:target}'.] If $\varphi$ is an atom then choose
  $\rho_1\to\cdots\to\rho_n\to\varphi\in\Gamma$ and call
  $\Gamma\vdash\rho_i$, for all $i=1,\ldots,n$.
\end{itemize}
We thus obtain the algorithm for inhabitation 
in the simply typed lambda-calculus known 
as the {\it Ben-Yelles algorithm\/}~\cite{BenYelles79}. 




The most important properties  
of WA
are the following.

\begin{lemma}\label{dudDFD}~ 
\begin{enumerate}\item\label{dudDFDO1}
The algorithm WA accepts an \ipc judgement 
\iff it is provable.
\item \label{dudDFDO2} 
All formulas occurring in any run of the algorithm are subformulas
of the formulas occurring in the initial judgement.
  \end{enumerate}
\end{lemma}
\begin{proof} (\ref{dudDFDO1}) We prove that  a judgement
$\Gamma\vdash\varphi$ is accepted 
if and only if there exists a~long normal form 
of type~$\varphi$ in~$\Gamma$. From left to right we proceed by
induction \wrt the definition of the algorithm,
using Lemma~\ref{trzyipolnazad}. In cases~(\ref{uwdiUDG}) 
and~(\ref{asudOFH}) the term $M$ is a~constructor, 
in case~(\ref{en:target}) it is an eliminator with
a head variable $x$ of type~$\psi$. For example, 
 if $\varphi=\tau\vee\sigma$
and $\psi=\gamma_1\to\gamma_2\to\alpha\vee\beta$ then
$M=xN_1N_2[z\:\alpha.\,P;\ v\:\beta.\,Q]$, where
$N_1, N_2, P, Q$ are long normal forms obtained in the
four recursive (or parallel) calls.

From right to left we work by induction \wrt the size of 
the lnf using Lemma~\ref{trzyipoltam}. 
For example, in the case of the term 
$xE_1\ldots E_m[u\:\alpha.\, P;\ v\:\beta.\, Q]$, 
types of $E_1,\ldots,E_m$ make a trace $T$ 
to $\alpha\vee\beta$ in $\psi$, and we can use induction
for $\Gamma,u\:\alpha\vdash P:\varphi$ and
$\Gamma,v\:\alpha\vdash Q:\varphi$.

(\ref{dudDFDO2})
In each of the steps of WA
  each new formula must
  be a subformula
  of either the present proof goal or one of the
  assumptions.
\end{proof}


\subsubsection*{Monotonic automata}

We define here a natural notion of automaton used as 
operational semantics of \ipc. 
This notion is a simplification
of the automata introduced by Barendregt, Dekkers 
and Schubert \cite{aleksybar} 
and of those used in~\cite{trudne} (but 
differs significantly from the notion introduced by 
Tzevelekos~\cite{Tzevelekos11}).

The idea is simple. If we read a proof task $\Gamma\vdash\varphi$
as a configuration of a machine, then any action taken by WA 
results in expanding the memory $\Gamma$ and proceeding
to a new internal state, yielding a new task (or a new 
configuration) $\Gamma'\vdash\varphi'$. For example, if 
an assumption of the form $(p\to q)\to r\in \Gamma$ is used to 
derive $\Gamma\vdash r$, then the next task $\Gamma,p\vdash q$
is a result of executing an instruction that can be 
written as
\relax{$r:\czek{(p\to q)\to r}{p}{q}$} (``in state~$r$
\relax{check the presence of $(p\to q)\to r$ in memory,}
add $p$ to the storage and go to state~$q$'').

A {\it monotonic automaton\/} is therefore 
defined as $\M = \<Q,R,f,\I\>$, where 

\begin{itemize}[noitemsep,nolistsep]
\item $Q$ is a finite set of states, with $f\in Q$ as the final state.
\item $R$ is a finite set of registers;
\item $\I$ is a finite set of instructions 
 of the form:
\begin{itemize}[noitemsep,nolistsep]
\item[(1)] 
$q: \czek{S_1}{S_2}{p}$, or 
\item [(2)] 
$q: \jmp{p_1}{p_2}$, 
\end{itemize}
where $q,$ $p,p_1,p_2\in Q$ and $S_1,S_2\subseteq R$.
\end{itemize}
We define a~{\it configuration\/} of~$\M$ as
a~pair $\<q,S\>$, where $q\in Q$ and $S\subseteq R$. 
Let $I\in \I$. 
The transition relation
$\<q,S\>\to_I \<p,S'\>$ holds 
\begin{itemize}[noitemsep,nolistsep]
\item for $I$ of type (1), when $S_1\subseteq S$, $S'=S\cup S_2$; 
\item for $I$ of type (2), when $S=S'$, and $p=p_1$ or $p=p_2$.
\end{itemize}
A~configuration $\<q,S\>$ is {\it accepting\/}
 when either $q=f$, or
\begin{itemize}[noitemsep,nolistsep]
\item $\<q,S\>\to_I \<p,S'\>$, where $I$ is of type (1), and 
$\<p,S'\>$ is accepting, or
\item  $\<q,S\>\to_I \<p_1,S\>$ and $\<q,S\>\to_I \<p_2,S\>$, 
where $I$ is of type (2), and
both $\<p_1,S\>$ and $\<p_2,S\>$ are \mbox{accepting}.
\end{itemize}

Observe that a monotonic automaton is an alternating machine.
Instructions of type (2) introduce
\relax{\emph{universal branching}},
and 
\relax{\emph{nondeterminism}}
occurs when more than one instruction is 
applicable in a~state.\footnote{But states themselves are not classified 
  as existential or universal.} %
The name ``monotonic'' is justified by 
the memory usage: registers are write-once devices, once raised (set to 1)
they cannot be reset to zero. Note also that all tests are positive: the
machine cannot see that a register is off.
A {\it nondeterministic\/} automaton
is one without universal branching (cf.~Section~\ref{sec:three-minus}).

It should be clear that our definition is motivated by proof search.
Indeed, the \relax{algorithm} WA is almost immediately 
implemented as an automaton.

\begin{proposition}\label{wddoDH32}
  Given a formula  $\Phi$ in \ipc, one can 
 construct (in logspace)
 a monotonic automaton $\M_\Phi$ and state~$q$ 
so that $\,\vdash\Phi$ \iff the configuration
  $\<q,\pusty\>$ of $\M_\Phi$ is~accepting.
\end{proposition}
\begin{proof} \relax{(Sketch)}
Let $S$ be the set of all subformulas of~$\Phi$.
Define \mbox{$\M = \<Q,R,f,\I\>$}, where

\begin{itemize}[noitemsep,nolistsep]
\item $R=S$ is the set of registers. 
\item The set of states $Q$ contains $S$ and some auxiliary states.
\end{itemize}

Under this definition, a \relax{judgement} $\Gamma\vdash\varphi$ 
corresponds directly to 
a~configuration $\<\varphi,\Gamma\>$ of $\M$. The instructions 
of the automaton now implement cases (\ref{uwdiUDG}--\ref{en:target}) 
of WA. Of course the following instructions are in~$\I$:

1. $\varphi: \jmp{\tau}{\sigma}$, ~~~~~~~~~~~~~~~~~ \ciut
for each $\varphi=\tau\wedge\sigma\in S$;

2. $\varphi: \czek{\pusty}{\tau}{\sigma}$, ~~~\, for each 
$\varphi=\tau\to\sigma\in S$;

Case (\ref{en:target}) of WA splits into three subcases handled with
help of auxiliary states, and depending on a choice of a formula
$\psi\in S$.

If $\varphi$
is an atom, $\varphi\in \TG\psi$, for some $\psi\in S$, and
$\{\rho_1,\ldots,\rho_m\}\in\trejs(\varphi,\psi)$, then $\I$ 
contains a~sequence of instructions (using $m-2$ brand new states)
abbreviated as:

3a. $\varphi: \czek\psi\pusty{\rho_1,\ldots,\rho_m}$\,.

If $\varphi$ is an atom or a disjunction, and $\bot\in\TG\psi$, 
for some $\psi\in S$, and 
$\{\rho_1,\ldots,\rho_m\}\in\trejs(\bot,\psi)$, 
then $\I$ also contains similar instructions:

3b. $\varphi: \czek\psi\pusty{\rho_1,\ldots,\rho_m}$\,.

If $\varphi$ is an atom or a disjunction, 
$\alpha\vee\beta
\in\TG\psi$, 
for some $\psi\in S$, and 
$\{\rho_1,\ldots,\rho_m\}\in\trejs(\alpha\vee\beta,\psi)$, then 
$\I$ contains instructions (using $m$ auxiliary states 
including $s_1$ and $s_2$):

3c. $\varphi: \czek\psi\pusty{\rho_1,\ldots,\rho_m,s_1,s_2}$;

~~~~~$s_1:\czek\pusty{\alpha}\varphi$;

~~~~~$s_2:\czek\pusty{\beta}\varphi$.

For example, 
if $\psi=\alpha\to\beta\vee\gamma\in\Gamma$,
and $\varphi\in S$ is 
an atom, then  
the following
instructions are in~$\I$ (where $p_1, p_2, p_3,p_4$
are fresh auxiliary states):
\begin{displaymath}
  \begin{array}{r@{\;:\;}l}
    \varphi
    & \czek{\psi}{\pusty}{p_1}\,;\\[1ex]
    p_1
    & \jmp{\alpha}{p_2}\,;\\[1ex]
    p_2
    & \jmp{p_3}{p_4}\,;\\[1ex]
    \,p_3
    &\czek{\pusty}{\beta}\varphi\,;\\[1ex]
    \,p_4
    &\czek{\pusty}{\gamma}\varphi\,.\\[1ex]
  \end{array}
\end{displaymath}
By straightforward induction one proves that a configuration
of the form $\<\varphi,\Gamma\>$ is accepting \iff $\Gamma\vdash
M:\varphi$ for some lnf $M$.
It remains to define $q$ as $\Phi$, and observe that
by~Lemma~\ref{dudDFD}(\ref{dudDFDO2}) the automaton can be 
computed in logspace. 
\end{proof}

\subsubsection*{Complexity}

The {\it halting problem\/} for monotonic automata is

\hfil ``{\it Given $\M,q,S$, 
is $\<q,S\>$ an accepting configuration of~$\M$?}\ciut''

In the remainder of this section we show that this problem 
is \relax{\pspace-complete.}
The upper bound is routine.


\begin{lemma}\label{uusuD4ERR}
It is decidable in polynomial space 
if a given configuration of  a monotone automaton is accepting. 
For nondeterministic automata (with no universal branching)
the  problem is in \np.
\end{lemma}

\begin{proof} 
An accepting computation of an alternating automaton can be seen as 
a tree. Every branch of the tree is a (finite or infinite)
sequence 
$\<q_0,S_0\>, \<q_1,S_1\>, \<q_2,S_2\>,\dots$ of configurations, 
where $S_0\subseteq
S_1\subseteq S_2\subseteq\cdots$. If the number of states and the 
number of registers are bounded by~$n$ then a configuration must
necessarily be repeated after at most $n^2$ steps. An alternating
Turing Machine working in time $n^3$ can therefore verify if 
a given configuration is accepting, and our halting problem 
is in \relax{$\mbox{\sc APtime}\subseteq \pspace$,}
cf.~\cite[Ch.~19]{Papadimitriou95}. In case of no universal branching,
a~nondeterministic Turing Machine suffices. 
\end{proof}

The next example hints on the technique used to show the 
lower bound. 

\begin{example}\label{edssheah}\rm
Consider a finite automaton $\A$, with states $\{0,\ldots,k\}$, 
the initial state~0,
and the final state~$k$. Given a string $w=a_1\ldots a_n$, we define 
a~monotonic~$\M$ such that $\A$ accepts~$w$ \iff 
the initial configuration $\<q^0,r^{0}_0\>$ of~$\M$ 
is accepting.\footnote{We write 
sets without $\{~\}$ whenever it is convenient.}

States of $\M$ are $q^0,q^1,\ldots ,q^n,f$ where $q^0$ is initial and $f$ is final. 
Registers are $r^t_i$, for $t\leq n$ and $i\leq k$.
For all $t = 0,\ldots ,n-1$, we have an instruction

\hfil $q^t: \czek{r^{t}_i}{r^{t+1}_j}{q^{t+1}}$,

whenever 
$\A$, reading $a_{t+1}$ in state $i$, can enter state $j$.
For $t=n$, we take
at last 

\hfil $q^n :\czek{r^n_k}{\pusty}{f}$.

Then an accepting computation of
the automaton 
$\A$, consisting of states
$0, i_1, i_2,\ldots,i_n=k$, is represented by a computation
of $\M$, ending in $\<f,r^0_0,r^1_{i_1},\ldots,r^{n}_{i_n}\>$.
Note that the instructions of $\M$ 
are all of type (1), i.e., there is no alternation.

{\it Correctness\/}: By induction \wrt $n-t$ one shows 
that a~configuration of the form
$\<q^t,r^0_0,\ldots,r^t_{i_t}\>$ is accepting \iff $\A$ accepts
the suffix $a_{t+1}\ldots a_n$ of $w$ from state~$i_t$.
\end{example}


In order to simplify the proof of \pspace-hardness,
let us
begin with the following observation. Every language 
\relax{$L\in \pspace$}
reduces in logarithmic space to some 
context-sensitive language~$L'$, recognizable by a linear bounded 
automaton (LBA), cf.~\cite[Ch.~9.3]{przasniczka}.
Indeed, for any  \mbox{$L\in\mbox{\sc Dspace}(n^k)$}, 
take  $L'=\{w\$^{n^k}\ |\ w\in L\wedge |w|=n\}$, where $|w|$ denotes
the length of the word~$w$. Hence it suffices to reduce the  
halting problem for LBA (aka {\sc In-place Acceptance} problem,
cf.~\cite[Ch.~19]{Papadimitriou95}) 
to the halting problem of monotonic automata. This retains 
the essence of \pspace
but reduces the amount of bookkeeping. 

Given a linear bounded automaton $\A$ and an input string
$w=x^1\ldots x^n$,
%
we construct a~monotonic automaton $\M$ and an initial configuration $C_0$
such that 

\hfil $\A$ accepts $w$ ~~~\iff~~~$C_0$ is an accepting configuration of~$\M$.

Let $p$ be a polynomial such that $\A$ works in time $2^{p(n)}$.
The alternating automaton $\M$ simulates $\A$ by splitting the $2^{p(n)}$ 
steps of computation recursively into halves and executing the 
obtained fractions concurrently. The ``history'' of each branch 
of the computation tree of $\M$ is recorded in its registers.
For every
\relax{$d=0,\ldots, p(n)$,}
there are three groups of registers
(marked $B,E,H$) representing $\A$'s configurations 
at the beginning ($B$) and at the end ($E$) of a~computation 
segment of up to $2^d$ steps, and halfway ($H$) through that segment.
\relax
That is, for any \relax{$i=1,\ldots, n$, $d=0,\ldots, p(n)$,}
for any state $q$ of $\A$, 
and for any tape symbol $a$ of $\A$, the automaton $\M$ has the 
following registers:


\newcommand{\regS}[3]{{s(#1,\, #2,\, #3)}\xspace}
\newcommand{\regC}[4]{{c(#1,\, #2,\, #3,\, #4)}\xspace}
\newcommand{\regH}[3]{{h(#1,\, #2,\, #3)}\xspace}

\begin{tabular}{ll}
  $\bullet$~~ & $\regS{B}{d}{q}$, $\regS{H}{d}{q}$, $\regS{E}{d}{q}$
   \quad --~~
    ``the current state of $\A$ is $q$'';\\[1ex]

  $\bullet$~~ &  $\regC{B}{d}{i}{a}$, $\regC{H}{d}{i}{a}$, $\regC{E}{d}{i}{a}$
    \quad --~~
    ``the symbol at position $i$ is $a$'';\\[1ex]

  $\bullet$~~ & $\regH{B}{d}{i}$, $\regH{H}{d}{i}$, $\regH{E}{d}{i}$
  \quad --~~
    ``the machine head scans position~$i$''.\\[1ex]

\end{tabular}

By $B_d$, $H_d$, $E_d$ we denote the sets of all registers 
indexed by $d$ and, respectively, by $B,H,E$. A set of registers 
$S\subseteq X_d$ is an 
$X,d$-{\it code\/} of a configuration
of $\A$, when $S$ contains exactly one register
of the form
\relax{$\regS{X}{d}{q}$}, exactly one $\regH{X}{d}{j}$, 
and, for every~$i$,  exactly one 
\relax{$\regC{X}{d}{i}{a}.$}

The initial configuration of $\M$ is $C_0=\<0,S_0\>$, where 
$S_0$ is the $B,p(n)$-code of the initial configuration of~$\A$, that
is,\\[0.5ex]
\begin{tabular}{@{}l@{\;\;}l}
  -- $ S=\{$
  &
    $\regS{B}{p(n)}{q_0},\;\;
    \regC{B}{p(n)}{1}{x^1}, \ldots, \regC{B}{p(n)}{n}{x^n},\;\;
    \regH{B}{p(n)}{1}\;\}$.

\end{tabular}

The machine $\M$ works as follows:

In the initial phase (commencing in state 0)
it guesses the final configuration of $\A$, and sets
the appropriate registers in $E_{p(n)}$ to obtain the $E,p(n)$-code of
that final configuration. Then $\M$ enters state
$Q_{p(n)}$. 
%

Assume now that the machine is in configuration $\<Q_d,S\>$, where
$d>0$, and $S$ contains:

 --  a $B,d$-code of some configuration~$C^b$ of~$\A$;

-- an $E,d$-code of some configuration~$C^e$ of~$\A$.




The following steps are now executed:

(1)  An intermediate configuration $C^h$ is guessed, i.e., registers in 
$H_d$ are nondeterministically set to obtain an $H,d$-code of $C^h$.
The machine selects an adequate sequence of instructions from the 
set below (where $q'$, $a$, and $j$ are arbitrary):
\begin{displaymath}
  \begin{array}{r@{\;:\;}ll}
    Q_d
    & \czek{\pusty}{\regS{H}{d}{q'}}{Q_d^{1}}; & \\[0.5ex]
    Q_d^{i}
    & \czek{\pusty}{\regC{H}{d}{i}{a}}{Q_d^{i+1}}, & \mbox{ for } i=1,\ldots, n; \\[0.5ex] 
    Q_d^{n+1}
    & \czek{\pusty}{\regH{H}{d}{j}}{Q'_d}. & \\
  \end{array}
\end{displaymath}

(2) The machine makes a universal split into states $Q^B_d$ and $Q^E_d$.

(3) In state $Q^B_d$ registers in $S\cap B_d$ are copied to
corresponding registers in $B_{d-1}$. This has to be done
nondeterministically, by guessing which registers in $S\cap B_d$
are set. The relevant instructions are: 
\begin{displaymath}
  \begin{array}{@{}r@{\;\;}l@{}}
    Q^B_d\;:
    & \czek{\regS{B}{d}{q}}{\regS{B}{d-1}{q}}{Q_d^{B,1}}; \\[0.5ex]
    Q_d^{B,i}\;:
    & \czek{\regC{B}{d}{i}{a}}{\regC{B}{d-1}{i}{a}}{Q_d^{B,i+1}},\quad\;\; 
      \mbox{
      for } i=1,\ldots, n; \\[0.5ex] 
    Q_d^{B,n+1}\;:
    & \czek{\regH{B}{d}{j}}{\regH{B}{d-1}{j}}{Q^{BE}_{d}}. \\
  \end{array}
\end{displaymath}

Then  registers in $S\cap H_d$ are copied to 
$E_{d-1}$ in a similar way. In short, this can be informally written as
$B_{d-1}:= B_d; E_{d-1}:=H_d$. Then the machine enters state~$Q_{d-1}$.

(4) In state $Q^E_d$, the operations follow  a~similar scheme, 
that can be written in short as
\begin{displaymath}
  B_{d-1}:= H_d; E_{d-1}:=E_d; ~{\sf jmp}~Q_{d-1}.  
\end{displaymath}

The above iteration splits the computation of~$\M$ into $2^{p(n)}$ 
branches, each eventually entering state~$Q_0$. At this point 
we verify the correctness. The sets 
$S\cap B_0$ and $S\cap E_0$ should now encode some configurations 
$C^b$ and $C^e$ of $\A$ 
such that either $C^b=C^e$, or $C^e$ is obtained from~$C^b$ in one
step. This can be nondeterministically verified,  and 
afterwards $\M$ enters its final state.

This last phase uses, in case $C^b=C^e$, the supply of instructions
below (the other variant can be handled similarly).
\begin{displaymath}
  \begin{array}{r@{\;:\;}ll}
Q_0
&  \czek{\regS{B}{d}{q}}\pusty{Q^s_d(q)}; & \\[0.5ex]
Q^s_d(q)
&  \czek{\regS{E}{d}{q}}\pusty{Q^{c,1}_d};& \\[0.5ex]
Q^{c,i}_d
& \czek{\regC{B}{d}{i}{a}}\pusty{Q^{c,i}_d(a)};& \\[0.5ex]
Q^{c,i+1}_d(a)
& \czek{\regC{E}{d}{i}{a}}\pusty{Q^{c,i+1}_d};& \\[0.5ex]
Q^{c,n+1}_d
& \czek{\regH{B}{d}{j}}\pusty{Q^h_d(j)};& \\[0.5ex]
Q^h_d(j)
& \czek{\regH{E}{d}{j}}\pusty{f}.& \\[0.5ex]
  \end{array}
\end{displaymath}

The main property of the above construction is the following.

\begin{lemma}\label{ueuEFUOE77}
Let $S$ be a set of registers such that:

--  $S\cap B_d$ is a $B,d$-code of some configuration~$C^b$ of~$\A$;

-- $S\cap E_d$ is an $E,d$-code of some configuration~$C^e$ of~$\A$.

In addition, assume that $S\cap H_d=\pusty$,
as well as $S\cap (B_e \cup H_e\cup E_e)=\pusty$, for all $e< d$. 
Then the following are equivalent:
\begin{enumerate}\item 
$\<Q_d,S\>$ is an accepting configuration of~$\M$;
\item $C^e$ is reachable from $C^b$ in at most $2^d$ steps of~$\A$.
\end{enumerate}
\end{lemma}
\begin{proof} (1) $\To$ (2): 
  Induction \wrt the definition of acceptance.

 (2) $\To$ (1): Induction \wrt~$d$. 
\end{proof}


\begin{theorem}\label{wdiugi44}
  The  halting problem for monotonic automata is \pspace-complete.
\end{theorem}
\begin{proof} 
  Lemma~\ref{uusuD4ERR} implies 
  that the problem  
  belongs to \pspace.
  The hardness part is a consequence 
of Lemma~\ref{ueuEFUOE77} applied for $d=p(n)$ with $C^b$ and $C^e$ 
being, respectively, the initial and final configuration of $\A$.
\end{proof}

\subsubsection*{Automata to formulas}

In order to finish our reduction of provability in IPC to
  provability in IIPC we need
to prove a specific converse of Proposition~\ref{wddoDH32}.  %
Consider a monotonic automaton $\M= \<Q,R,f,\I\>$, and 
an ID of the form $C_0=\<q_0,S_0\>$. Without loss of generality 
we can assume that \mbox{$Q\cap R = \pusty$.}
Using
states and registers of $\M$ as propositional atoms,
we define a~set
of axioms $\Gamma$ so that $\Gamma \vdash q_0$ \iff 
$C_0$ is accepting. The set $\Gamma$ contains the atoms
$S_0\cup\{f\}$;
other axioms in~$\Gamma$ represent instructions of~$\M$.

An axiom for
$q: \czek{S_1}{S_2}{p}$, where
$S_1=\{s^1_1,\ldots,s^k_1\}$ and $S_2=\{s^1_2,\ldots,s^\ell_2\}$, is:

(1) \hfil $s^1_1\to\cdots\to s^k_1\to (s^1_2\to\cdots \to s^\ell_2\to p)\to q$.

And for every instruction $q: \jmp{p_1}{p_2}$, 
 there is an axiom 

(2) \hfil $p_1\to p_2 \to q$.

Observe that all the above axioms are of order at most two, hence
the formula $\Gamma\to q_0$ has order at most three.


\begin{lemma}\label{dtfovfe7}
  Given
  the above definitions,
the configuration~$\<q_0,S_0\>$ is accepting \iff $\Gamma\vdash q_0$ holds.
\end{lemma}

\begin{proof}
For every $S\subseteq R$ and $q\in Q$,  we prove that 
$\Gamma,S\vdash q$ \iff $C=\<q, S\cup S_0\>$ is accepting. 
We think of 
$\Gamma$
as of a type environment where each 
axiom is a declaration of a~variable. 

\relax{$(\Leftarrow)$} Induction \wrt long normal 
proofs. With $\to$ as the only connective,
any normal proof $T$ of an atom $q$ must be 
a variable or an application, 
say $T=xN_1\ldots N_m$, 
The case of $x\:f$ (i.e., $q=f$) is
obvious; otherwise the type of $x$ corresponds 
to an instruction. There are two possibilities:

(1) If $x \: s^1_1\to\cdots\to s^k_1\to (s^1_2\to\cdots\to  s^\ell_2\to p)\to q$,

then actually we obtain that
%
$T=xD_1\ldots D_k(\lambda u_1\:s^1_2\ldots\lambda u_\ell\:s^\ell_2.\,P)$.
Terms $D_1,\ldots,D_k$ are, respectively, of types $s^1_1,\ldots,s^k_1$, 
and they must be variables declared in~$S$,
as there are no  other assumptions with targets
$s^1_1,\ldots,s^k_1$.
Hence the
instruction corresponding to~$x$ is applicable at $C=\<q,S\>$
and yields $C'=\<p,S\cup S'\>$,
where $S'=S\cup \{s^1_2,\ldots,s^\ell_2\}$.
In addition we have
$\Gamma,S\cup S'\vdash P:p$, 
whence
$C'$ is accepting  by the induction hypothesis.

(2) If $x$ has type 
$p_1\to p_2 \to q$,
where $p_1, p_2\in  Q$,

then $T=xT_1T_2$.
The appropriate universal instruction leads to two
IDs:
$C_1=\<p_1,S\>$ and $C_2=\<p_2,S\>$. The \relax{judgements}
$\Gamma,S\vdash T_1:p_1$ and $\Gamma,S\vdash T_2:p_2$ 
obey the induction hypothesis. Thus $C_1, C_2$ are accepting
and so
is~$C$.

$(\To)$ Induction \wrt the definition of acceptance. 
\end{proof}


\begin{proposition}\label{eyfeyg022}
The halting problem for monotonic automata reduces to the 
provability problem \relax{}
for formulas in \iipc of order at most three.
\end{proposition}

Putting together Propositions~\ref{wddoDH32} and~\ref{eyfeyg022} 
and Theorem~\ref{wdiugi44}
we obtain a number of consequences. 

\begin{theorem}\label{QEA2552}
  Provability in \ipc and \iipc are
  \pspace-complete.
\end{theorem}

While the statement of Theorem~\ref{QEA2552} is well-known~\cite{star79},
the present automata-theoretic proof directly exhibits the computational
capability of proof-search: polynomial space Turing Machines are encoded
almost directly into propositional formulas (the automaton being merely 
a technically convenient way to manipulate them). 

Without loss of generality we can interpret problems in
\pspace
as reachability questions concerning some objects or configurations
of polynomial size. The construction used in the proof of 
Theorem~\ref{wdiugi44} (the simulation of LBA) 
reflects a natural, possibly interactive, approach to solve such
questions: split the reachability task into two, by choosing some 
intermediate configuration.  An example that comes to mind is 
the Sokoban game: the set of winning positions is a~context-sensitive 
language and one can try to solve the puzzle by determining some 
milestone states. 

Another consequence of our development 
is that the computational power of \ipc is fully 
contained in \iipc, and in an apparently small fragment.


\begin{theorem}\label{redran3}
  For every formula $\varphi$ of full \ipc 
one can
  construct (in logspace) an implicational formula $\psi$ of order at most three
such that 
$\psi$ is provable iff so is~$\varphi$.
\end{theorem}

\section{An intuitionistic order hierarchy}\label{sec:hierarchy}

In Section~\ref{monaut}, we observed that provability in the whole \ipc is
faithfully reflected by provability for formulas of \iipc of that have
order at most three. Proving any formula is therefore at most as difficult
as proving some formula of order three. But is every formula {\it equivalent\/}
to one of order~three? The answer is negative:
in the case of 
\ipc we have a strict order hierarchy of formulas. 
Define by induction $\varphi^1=p_1$ and $\varphi^{k+1}=\varphi^k\to p_{k+1}$.
That is, 

\hfil $\varphi^k=(\cdots((p_1\to p_2)\to p_3)\to\cdots\to p_{k-1})\to p_k$.


\begin{lemma}\label{lemataleksego}
  Every proof of $\varphi^k\to\varphi^k$ is $\beta\eta$-convertible
  to  the identity combinator $\lambda x.x$.
\end{lemma}
\begin{proof}
We prove  the following generalized statement by induction \wrt~
the number~$k$.
Let 
$\tg\gamma\not\in\{p_1,\ldots,p_k\}$, for
all $\gamma\in\Gamma$, 
and let $\Gamma, X\:\varphi^k\vdash M:\varphi^k$, where
$M$ is in normal form. Then $M=_{\beta\eta}X$.
Indeed, first note that $X$ is the only assumption with target~$p_k$,
hence for $k=1$ the claim follows immediately. Otherwise either
 $M=X$ or $M=\lambda Y.\,M'$ with a derivation
$\Gamma, X\:\varphi^k, Y\:\varphi^{k-1}\vdash M':p_k$.
This is only possible when 
$M'=XM''$, where 
\mbox{$\Gamma, X\:\varphi^k, Y\:\varphi^{k-1}\vdash M'': \varphi^{k-1}$}.
By the induction hypothesis
 for
$\Gamma' = \Gamma \cup \{ x:\varphi^k\}$
and $Y:\varphi^{k-1}$,
we have $M''=_{\beta\eta}Y$, hence
$M= \lambda Y.\,XM'' =_{\beta\eta}\lambda Y.\,XY=_{\beta\eta} X$.
\end{proof}

\begin{theorem}\label{hierimp}
For every $k$, no implicational formula of order  less than~$k$ is 
intuitionistically equivalent to $\varphi^k$. 
\end{theorem}
\begin{proof}
  If $\vdash \varphi^k\oto \alpha$
  then there are closed terms 
$T:\varphi^k\to \alpha$ and $N:\varphi^k\to \alpha$. The composition
$\lambda x.\, N(T x)$ is a combinator of type $\varphi^k\to \varphi^k$,
and by Lemma~\ref{lemataleksego} it must be \mbox{$\beta\eta$-equivalent} 
to identity. That is, $\varphi^k$ is a {\it retract\/} of $\alpha$, 
in the sense of~\cite{ru-retrakcje02}. It thus follows 
from~\cite[Prop.~4.5]{ru-retrakcje02} that~$\alpha$ must be of order
at least~$k$.
\end{proof}

Interestingly enough, Theorem~\ref{hierimp} 
stays in contrast with the
situation in classical logic.  Every propositional formula is
classically equivalent to a formula in conjunctive normal form
(CNF). If implication is the only connective then we have a similar
property.

\begin{proposition}\label{klas3}
Every implicational formula is classically equivalent to a formula
of order at most~three.
\end{proposition}
\begin{proof}
Given a formula of the form $\varphi=\alpha_1\to\cdots\to\alpha_n\to p$, 
we first 
translate the conjunction $\alpha_1\wedge\cdots\wedge\alpha_n$ into 
a conjunctive normal form $\beta_1\wedge\cdots\wedge\beta_m$, so that $\varphi$
is equivalent to
a formula 
$\psi=\beta_1\to\cdots\to\beta_m\to p$. Each $\beta_i$ 
is a disjunction of literals. For every $i$, there are two possibilities:

Case 1: At least one component of
$\beta_i$ is a 
variable, say 
$\beta_i=\neg q_1\vee\cdots\vee \neg q_r\vee r_1\vee\cdots\vee r_k\vee s$.
We replace $\beta_i$ in $\psi$ by the formula
$\beta'_i= q_1\to\cdots\to q_r\to (r_1\to p)\to\cdots\to (r_k\to p)\to s$.

Case 2: All components of
the formula 
$\beta_i$ are negated variables, 
i.e.,~$\beta_i = \neg q_1\vee\cdots\vee \neg q_r$. Then we replace such 
$\beta_i$ by
the formula 
$q_1\to\cdots\to q_r\to p$.

 For example, if 
$\psi = (s\vee q\vee \neg r)\to (\neg q\vee \neg r\vee \neg s)\to p$
then we rewrite $\psi$ as the formula
$(r\to (q\to p)\to s)\to (q\to r\to s\to p)\to p$. 
It is a routine exercise
to see that  the final result is a formula of rank at most 3 which is
classically equivalent to~the initial formula~$\varphi$ (note that 
if a~Boolean
 valuation falsifies~$p$ then it~satisfies $p\oto\bot$).
\end{proof}

\begin{example}\label{abcde}\rm
The formula $\varphi^5=(((p_1\to p_2)\to p_3)\to p_4)\to p_5$ 
is classically equivalent to this ``normal form'': 
$$
(p_1 \to (p_2 \to p_5) \to p_4) \to (p_3 \to p_4) \to p_5.
$$
\end{example}


\begin{remark}\rm Despite the contrast between the classical CNF collapse
and order hierarchy in intuitionistic logic,
there is still a strong analogy between CNF and order three fragment of \iipc.
The CNF formulas do indeed exhaust the whole expressive
power of classical propositional logic, but for a heavy price.
The value-preserving translation of a formula 
to conjunctive normal form is exponential, hence useless \wrt{} %
\np-completeness of CNF-SAT. That requires 
a polynomial satisfiability-preserving translation, very much 
like our provability-preserving reduction of the full \ipc to 
order three \iipc. 
\end{remark}

\section{Below order three} 
\label{sec:low-rank}

In this section we identify fragments of \iipc corresponding to the 
complexity classes P, \relax{\np and \conp.}


\subsection{Formulas of order two: deterministic polynomial time }

Implicational formulas of rank 1 are the same as propositional clauses
in logic programming. Therefore decision problem for rank 2 formulas
(no matter, classical or intuitionistic) amounts to propositional
logic programming, 
known to be P-complete~\cite{dantsin} 
\wrt logspace reductions.


\subsection{Order three minus: class \np}
\label{sec:three-minus}

We define here a subclass of \iipc 
for which the provability problem is \relax{\np-complete.}

We split the set $\tyv$ of propositional variables into \relax
 two disjoint infinite subsets  
$\tyv_0, \tyv_1\subseteq \tyv$, called respectively {\it data\/}
and {\it control\/} variables.
The role of control variables is to
occur as targets, the data variables only occur as arguments. 
The set of formulas $\Types_{3-}$ is defined by the 
grammar:
\begin{displaymath}
  \begin{array}{ll}
    \Types_{3-} ::= \tyv_1\ |\ \Types_{2-}\to\Types_{3-}\ |\ 
\tyv_0\to\Types_{3-}\\
    \Types_{2-} ::= \tyv_1\ |\ \tyv_0\to\Types_{2-} \mid 
 \Types_{1-}\to\Types_{1-} \\
    \Types_{1-} ::= \tyv_1 \mid \tyv_0 \to \Types_{1-}\\
  \end{array}
\end{displaymath}

Formulas in 
$\Types_{1-}$ are of the form $p_1\to\cdots\to p_n\to q$,
where $p_i\in\tyv_0$ and $q\in\tyv_1$. The set $\Types_{2-}$
consists of formulas of order at most two, with an $\tyv_1$ target,
and with at most one argument in~$\Types_{1-}$, and all other
arguments being variables in $\tyv_0$. Finally the $\Types_{3-}$
formulas are of shape $\sigma_1\to\sigma_2\to\cdots\to\sigma_n\to q$,
where $q\in\tyv_1$ and $\sigma_i\in\Types_{2-}\cup\tyv_0$,
for $i=1,\ldots,n$.

\begin{lemma}
  \label{lemma:linearinnp}
  Proof search for formulas in $\Types_{3-}$ is in \np.
\end{lemma}
\begin{proof} 
Proving an implicational formula amounts to proving its target
in the context consisting of all its arguments. In the case of 
$\Types_{3-}$
we are interested in contexts built from atoms in~$\tyv_0$ and formulas
in $\Types_{2-}$ (some of those can be atoms in $\tyv_1$). Such contexts
\relax{are}
called \np-{\it contexts\/}.
If $\Gamma$ is an \np-context, and $\Gamma\vdash M:q$, with $M$ an lnf, 
then $M$ is either a variable or it has  the form 
\mbox{$M=XY_1Y_2\ldots Y_k(\lambda V_1\ldots V_m.\,N)Z_1\ldots Z_\ell$},
where:

-- the type of~$X$ is a $\Types_{2-}$ formula
of the form
$$
  p_1\to p_2\to\cdots\to p_k\to \alpha\to p'_1\to\cdots\to p'_\ell\to q;
$$

-- $Y_1\:p_1, Y_2\:p_2,\ldots, Y_k\:p_k, Z_1\:p'_1,\ldots, Z_\ell\:p'_\ell$
are declared in $\Gamma$;

-- the term $\lambda V_1\ldots V_m.\,N$ has a $\Types_{1-}$ type 
 $\alpha = s_1\to\cdots\to s_m\to q'$. 

We then have $\Gamma, V_1\:s_1,\ldots, V_m\:s_m\vdash N:q'$, with
$s_1,\ldots, s_m\in\tyv_0$ and  $q'\in\tyv_1$, and the context
$\Gamma, V_1\:s_1,\ldots, V_m\:s_m$ is an \np-context. 
In terms of a monotonic automaton this proof construction step amounts
to executing this instruction:
$$
  q: \czek{p_1,\ldots,p_k,p'_1,\ldots,p'_\ell}{s_1,\ldots, s_m}{q'}
$$

No other actions need to be performed by the automaton except a~final 
step, which takes up the form 
$q:\czek{q}{\pusty}{f}$,
where $f$ is a~final
state (this corresponds to the case of $M=X$). 

It follows that $\Types_{3-}$ proof search can be handled by 
a nondeterministic automaton (with no universal branching). 
By Lemma~\ref{uusuD4ERR} provability in $\Types_{3-}$ belongs to \np.
%
%
%
\end{proof}

\begin{remark} \rm To exclude universal branching, only one argument in 
a~$\Types_{2-}$ formula can be non-atomic. Note however that 
formulas used in the proof of Proposition~\ref{eyfeyg022} satisfy
a similar restriction. Hence 
the separation between ``data atoms'' $\tyv_0$ and ``control atoms'' 
in $\tyv_1$ is essential too.

Similarly, sole separation between ``data atoms'' and
  ``control atoms'' does not reduce the complexity either, as it
  directly corresponds to separation between registers 
  and states of automata.
%
\end{remark}


\newcommand{\nnn}{{n}}

\begin{lemma}
  \label{lemma:linearnphard}
  Provability in $\Types_{3-}$  is   \np-hard.
\end{lemma}
\begin{proof}
  We reduce the 3-CNF-SAT problem to provability in $\Types_{3-}$.
For every 3-CNF formula 
%
%
$$
        \Psi=(\lir_{11}\vee\lir_{12}\vee\lir_{13})\wedge \cdots \wedge
(\lir_{k1}\vee\lir_{k2}\vee\lir_{k3}), \qquad\quad (*)\relax{}
$$
where $\lir_{ij}$ are literals, we construct a $\Types_{3-}$ formula~$\psi$
so that \relax{$\Psi$} is  classically satisfiable \iff $\psi$ has a proof. 
\relax{Assume that $\{p_1,\ldots,p_{\nnn}\}$ are} all propositional variables occurring in \relax{$\Psi$},
and 
that $p_1,\ldots,p_{\nnn},p'_1,\ldots,p'_{\nnn}\in \tyv_0$.
Other atoms of
the formula 
$\psi$ are 
$q_1,\ldots,q_{\nnn},c_1,\ldots,c_{k}\in \tyv_1$.

Define $\rho_{ij}=p_\ell$ when $\lir_{ij}=p_\ell$, and
$\rho_{ij}=p'_\ell$ when $\lir_{ij}=\neg p_\ell$. 
The formula $\psi$ has the form $\Gamma\to q_1$, where $\Gamma$ consists
of the following axioms:

\begin{enumerate}
\item\label{kudfgI11}
 $(p_i \to q_{i+1}) \to q_i$ and $(p'_i \to q_{i+1}) \to q_i$, for 
$i=1,\ldots,\nnn-1$;\relax{}
\item\label{kudfgI22} $(p_{\nnn} \to c_1) \to q_{\nnn}$ and
  $\ (p'_{\nnn} \to c_1) \to q_{\nnn}$; \relax{}
\item\label{kudfgI33} $\rho_{i1}\to\ c_{i+1}\to c_i$, 
$\ \rho_{i2}\to c_{i+1}\to c_i$, and $\ \rho_{i3}\to c_{i+1}\to c_i$,
for $i=1,\ldots, k-1$;\relax{}
\item\label{kudfgI44}
  $\rho_{k1}\to c_{k}$,
  $\ \rho_{k2}\to c_{k}$, 
  and $\ \rho_{k3}\to c_{k}$.\relax{}
\end{enumerate}
For a binary valuation~$v$, 
let $\Delta_v$ be such that $p_i\in\Delta_v$ when $v(p_i)=1$ and
$p'_i\in\Delta_v$ otherwise. Suppose that \relax{$\Psi$} is satisfied by 
some~$v$. Then, for every $i$ there is $j$
with $\rho_{ij}\in \Delta_v$ and one can readily see that
$\Gamma,\Delta_v\vdash c_1$ using axioms (\ref{kudfgI44}) and (\ref{kudfgI33}).

Let $\Delta^i_v=\Delta_v\cap(\{p_j\ |\ j < i\}\cup\{p'_j\ |\ j < i\})$. 
Since  $\Gamma,\Delta_v\vdash c_1$ we obtain
$\Gamma,\Delta^{\nnn}_v\vdash q_{\nnn}$\relax{}
using (\ref{kudfgI22}), and then we use (\ref{kudfgI11}) to 
prove $\Gamma,\Delta^i_v\vdash q_i$ by induction, for
$\nnn-1\geq i\geq 1$.\relax{}
Since $\Delta^1_v=\pusty$, we eventually get $\Gamma\vdash q_1$.

For the converse, a proof of $\Gamma\vdash q_1$ in long
  normal form must begin with 
a~head variable of type $(p_1 \to q_2) \to q_1$ or
$(p'_1 \to q_2) \to q_1$ applied to an abstraction
$\lambda x.\,N$ with $N$ of type~$q_2$. The shape of $N$ is also 
determined by axioms (\ref{kudfgI11}--\ref{kudfgI22}), and it 
must inevitably contain a proof
 of  $\Gamma,\Delta_v\vdash c_1$ for some $v$.
Such a proof is only possible if each of the $k$ clauses
is satisfied by~$v$. The fun of checking the details is left to the reader.
\end{proof}

We can put together Lemma~\ref{lemma:linearinnp} and
Lemma~\ref{lemma:linearnphard} to obtain the conclusion of this section:
a very limited fragment of \iipc is of the same expressive power as SAT.
\begin{theorem}\label{jadiudgPwq}
  Proof search for $\Types_{3-}$ formulas 
is \np-complete.  
\end{theorem}

\subsection{Order two plus}

We distinguish another natural class of formulas of low order 
for which the provability problem is \conp-complete.
We consider here implicational formulas built from
literals, and 
we restrict attention to formulas of order two,
where all literals are counted as of order zero.
We call this 
fragment {\it order two plus\/}. Note that if we use the standard
definition of order, these formulas are of order three.

It is convenient and illustrative 
to work with literals (using negation), but 
formulas of order two plus make in fact a fragment of \iipc.
Indeed, $\neg p = p\to\bot$ by definition, and in all our
proofs below the constant $\bot$ can be understood merely as 
a distinguished atom with no particular meaning. In other
words, the {\it ex falso\/} rule, i.e., $\bot$-elimination
is not involved.

\begin{lemma}\label{aduoQSada}
Formulas of order two plus have the linear size model property:
if $\,\nvdash\!\varphi$ then there is a~Kripke model of depth at most 
2 and of cardinality not exceeding the length of~$\varphi$.
\end{lemma}

\begin{proof}
Let $\varphi=\xi_1\to\cdots\to \xi_n\to\lip$, where 
$\xi_i = \liq^1_{\,i}\to\cdots \to \liq^{n_i}_{\,i}\to \lir_i$. 
Without loss of generality we may 
assume that
literals 
$\lip, \lir_1, \ldots, \lir_n$
are all either propositional variables or~$\bot$. 
Suppose that $\nvdash\varphi$. Then there exists a finite
Kripke model $C$ and a state $c_0$ of $C$ such that 
$C,c_0\notforces \varphi$. That is, $C,c_0\forces \xi_i$, for 
all $i=1,\ldots,n$, and $C,c_0\notforces \lip$. For every $i=1,\ldots,n$
we now select a~final state $c_i$ of~$C$ as follows. 
Since $C,c_0\forces \xi_i$, there are two possibilities: either
$C,c_0\forces \lir_i$, or
$C,c_0\notforces \liq_{\,i}^j$, for some~$j$.
The important case is when $C,c_0\notforces \liq^j$ and 
$\liq^j = \neg s$, for some propositional variable $s$.
Then there is a successor state $c'$ of $c_0$ with $C,c'\forces s$,
hence there also exists a final state forcing~$s$. We define $c_i$
as one of such final states. In other cases the choice of $c_i$
is irrelevant and we can choose any final state.

Now define a new model $C'$ with the set of states 
$\{c_0\}\cup\{c_1,\ldots,c_n\}$ and the relation 
$\forces$ inherited from~$C$,
i.e., $C',c\forces s$ iff $C,c\forces s$,
for any state $c$ of $C'$ and any propositional variable~$s$. Note
that so defined $C'$ has depth \relax{at} most 2.

We claim that $C',c_0\notforces \varphi$. Clearly $C',c_0\notforces \lip$,
so we should prove that all states in $C'$ force all formulas $\xi_i$.
Forcing in any state only depends on its successor states, hence
if we had $C,c_i\forces \xi_i$ then 
we still have $C',c_i\forces \xi_i$, for all $i=1,\ldots,n$, because 
nothing has changed at the final states. But also nothing has changed
at $c_0$ with respect to $\xi_i$. Indeed, if $C,c_0\forces\lir_i$
then $C',c_0\forces\lir_i$, and if $C,c_0\notforces \liq_{\,i}^j$ 
for some $j$, where~$\liq_{\,i}^j$ is a propositional variable,
then $C',c_0\notforces \liq_{\,i}^j$ as well.
Otherwise, for some $s,j$, we have $\neg s=\liq_{\,i}^j$ and
$C',c_i\forces s$, so $C',c_0\notforces \liq_{\,i}^j$.
\end{proof}

\begin{example}\rm Lemma~\ref{aduoQSada} cannot be improved to 
2-state models: the formula\\[0.5ex]
\mbox{~}\hfil $
(\neg p\to  q)\to  (\neg r\to  q)\to  (p\to  \neg r)\to  q  
$\\[0.5ex]
requires a countermodel with at least 3 states. 
\end{example}

\begin{theorem}\label{jadiudgQwq}
Order two plus fragment of \ipc is \conp-complete.
\end{theorem}

\begin{proof}
  That the problem is in \conp
  follows from Lemma~\ref{aduoQSada}: the small
countermodel can be guessed and verified in polynomial time. 

The \conp-hardness of order two plus is shown by a reduction from 
non-3-CNF-SAT. Let us begin with a formula in 3-CNF:

\hfil $\Psi = (\lir_{11}\vee\lir_{12}\vee\lir_{13})\wedge \cdots \wedge
(\lir_{k1}\vee\lir_{k2}\vee\lir_{k3})$,


\relax{where $\lir_{ij}$ are literals.}
Assume that $\{p_1,\ldots,p_n\}$ are all propositional variables 
in~$\Psi$. 
We define a~set $\Gamma_\Psi$ of formulas using propositional variables 
\relax{$p_1,\ldots, p_n, p'_1\ldots, p'_n$}.
For any literal $\lir_{jm}$ occurring 
in $\Psi$ we write \relax{$\lir'_{jm}$} to denote:

-- the variable \relax{$p'_i$}, when $\lir_{jm}=p_i$;

-- the variable ${p}_i$, when $\lir_{jm}= \neg p_i$.

Members of $\Gamma_\Psi$ are as follows
(for all $i=1,\ldots,n$ and  $j=1,\ldots,k$):

-- $X_i \: \neg p_i\to \neg p'_i\to\bot $; \relax{}

-- $Y_j \: \lir'_{j1}\to \lir'_{j2}\to \lir'_{j3}\to\bot$.\relax{}

For example, if the first component of $\Psi$ was $(p\vee\neg q\vee\neg s)$
then $Y_1:
p'\to q\to s\to\bot$.\relax{}
We shall prove that:
\begin{quote}
\hfil $\Psi$ is classically unsatisfiable ~\iff ~ $\Gamma_\Psi\vdash \bot$.  
\end{quote}


$(\To)$~ Let $m\leq n$ and let ${v}$
be a Boolean valuation 
of variables $p_1,\ldots,p_m$. Define an environment $\Gamma_{v} = 
\Gamma_\Psi\cup\{x_1\:  
p_1^{{v}},\ldots x_m\: p_m^{{v}}\}$,
\relax{}
where 

\hfil $ p_i^{{v}}=\przypadki{ p_i}{${v}(p_i)=1$}{p'_i}$ \relax{}

By a reverse induction~\wrt~$m$ we 
prove that 
$\Gamma_{v}\vdash\bot$, for every such~${v}$. 
We~begin with $m=n$. Then ${v}$ is
defined on all variables in $\Psi$ and does not satisfy~$\Psi$.
Therefore the value under ${v}$ of 
at least one clause  $\lir_{j1}\vee\lir_{j2}\vee\lir_{j3}$
is zero, in which case we have 
$\lir'_{j1}, \lir'_{j2}, \lir'_{j3}\in\Gamma_{v}$,\relax{}
hence $\bot$ is derivable using the assumption~$Y_j$. (For example, if
the unsatisfied component of $\Psi$ were $(p\vee\neg q\vee\neg s)$  
then we would have $p^{v}= p'$, 
$q^{v}= q$, $s^{v}= s$.)\relax{}

For the induction step assume the claim holds for some $m\leq n$, and let 
${v}$ be a valuation of $p_1,\ldots,p_{m-1}$. For $b=0,1$, 
define 
${v}_b$ as ${v}$ extended by ${v}_b(p_m)=b$.
By the induction hypothesis there are proofs $\Gamma_{{v}_0}\vdash
M_0:\bot$ and $\Gamma_{{v}_1}\vdash M_1:\bot$. 
Then one proves $\bot$ from $\Gamma_{v}$ 
using the assumption $X_m$; the lambda term in question has the form
$X_m(\lambda x_m\:p_m.\,M_1)
(\lambda x_m\: p'_m.\,M_0)$.\relax{}

$(\Ot)$~ By contraposition suppose that
  $v$ satisfies
  $\Psi$.  We extend it to primed propositional variables by letting
  $v(p') = 1- v(p)$.  Since
  $v$ satisfies all the clauses
  $\lir_{11}\vee\lir_{12}\vee\lir_{13}$ of $\Psi$,
  it satisfies all the formulas in $\Gamma_\Psi$. Consequently,
  $\Gamma\not\vdash\bot$ even in classical logic.

For any 
given Boolean 
valuation ${v}$ of $p_1,\ldots,p_n$, we 
prove that ${v}$ does not satisfy~$\Psi$. 
Let again $\Gamma_{v}= 
\Gamma_\Psi\cup\{x_1\: p_1^{{v}},\ldots, 
x_n\:p_n^{{v}}\}$.\relax{}
Since $\Gamma_\Psi \vdash\bot$, also $\Gamma_{v} \vdash\bot$, so let $M$
be the shortest possible normal lambda-term such that 
$\Gamma_{v} \vdash M:\bot$. The proof must begin with either some~\relax{$X_i$}
or some $Y_j$. In the first case it must be of the form
$M= X_i(\lambda y_i\:p_i.\,M_1)%
       (\lambda y_i\: p'_i.\,M_0)$,\relax{}
where $\Gamma_{v}, y_i\:,p_i\vdash M_1:\bot$ 
and $\Gamma_{v}, y_i\: p'_i\vdash M_0:\bot$.
But in $\Gamma_{v}$ we have either 
$x_i\:p_i$ or $x_i\: p'_i$.\relax{} Thus either 
$M_1[y_i:=x_i]$ or $M_0[y_i:=x_i]$ 
makes a proof of~$\bot$ shorter than~$M$.

It follows that the shortest proof of~$\bot$ is of the form
$M=Y_jN_1N_2N_3$, where $\Gamma_{v}\vdash N_1: \lir'_{j1}$, 
$\Gamma_{v}\vdash N_2 :\lir'_{j2}$, and 
$\Gamma_{v}\vdash N_3 :\lir'_{j3}$. Then $N_1, N_2, N_3$ must
be variables declared in $\Gamma_{v}$ which is only possible when
the literals $\lir_{j1}, \lir_{j2}, \lir_{j3}$ are zero-valued under~${v}$.
\end{proof}

\section{Conclusions and further research}
\label{sec:conclusions}

We have demonstrated the strength of implicational intuitionistic
propositional logic (\iipc) as a reasonable language
to express problems
solvable in \pspace.
%
%
%
%
%
%
Moreover, some natural subclasses of \iipc, called {\it order three minus\/}
and {\it order two plus\/}, 
correspond respectively to 
complexity classes \np and \conp
(Section~\ref{sec:low-rank}). 



\relax{ The situation in \iipc can be related to the one in modal
  logic S4 through the standard embedding
  \cite{TarskiMcKinsey48} (see~\cite{Bou04} for a modern account
  of the embedding). Each subsequent order corresponds through
  this embedding to one application of the modal operator. 
In particular, formulas of order three in \iipc translate to
formulas of modal depth
four. 
Interestingly enough, satisfiability for
  S4 formulas already of modal depth $k\geq 2$ is \pspace-complete
  \cite[Theorem 4.2]{Halpern95}.


\bibliographystyle{asl}
\bibliography{aspukz}

\end{document}